\documentclass[12pt,reqno,a4paper]{amsart}

\PassOptionsToPackage{hyphens}{url}
\usepackage[colorlinks,plainpages,hypertexnames=false,plainpages=false]{hyperref}
\hypersetup{urlcolor=blue, citecolor=blue, linkcolor=blue}
\usepackage{upgreek}
\usepackage[latin1]{inputenc}
\usepackage{amsmath}
\usepackage{amsthm}
\usepackage{amssymb}
\usepackage{amsfonts}
\usepackage{enumerate}
\usepackage{stmaryrd}
\usepackage{enumitem}
\usepackage{centernot}
\usepackage{mathrsfs}
\usepackage{listings}
\usepackage{multicol}
\usepackage{color}
\usepackage{array}
\newcolumntype{L}[1]{>{\raggedright\let\newline\\\arraybackslash\hspace{0pt}}m{#1}}
\newcolumntype{C}[1]{>{\centering\let\newline\\\arraybackslash\hspace{0pt}}m{#1}}
\newcolumntype{R}[1]{>{\raggedleft\let\newline\\\arraybackslash\hspace{0pt}}m{#1}}

\usepackage{tikz}
\usetikzlibrary{arrows}
\usetikzlibrary{decorations.pathreplacing}
\usetikzlibrary{matrix}
\usetikzlibrary{calc}
\usetikzlibrary{shapes}
\usetikzlibrary{patterns}
\usetikzlibrary{fit,backgrounds,scopes}

\newcommand\irregularcircle[2]{
  \pgfextra {\pgfmathsetmacro\len{(#1)+rand*(#2)}}
  +(0:\len pt)
  \foreach \a in {10,20,...,350}{
    \pgfextra {\pgfmathsetmacro\len{(#1)+rand*(#2)}}
    -- +(\a:\len pt)
  } -- cycle
}

\makeatletter
\DeclareRobustCommand{\rvdots}{%
  \vbox{
    \baselineskip4\p@\lineskiplimit\z@
    \kern-\p@
    \hbox{.}\hbox{.}\hbox{.}
  }}
\makeatother

\newtheoremstyle{theoremstyle}
{10pt}      
{5pt}       
{\itshape}  
{}          
{\bfseries} 
{}         
{\newline}      
{}          

\newtheoremstyle{algorithmstyle}
{10pt}      
{5pt}       
{}  
{}          
{\bfseries} 
{}         
{ }      
{}          

\newtheoremstyle{examplestyle}
{10pt}      
{5pt}       
{}          
{}          
{\bfseries} 
{}         
{\newline}      
{}          

\theoremstyle{theoremstyle}
\newtheorem{theorem}{Theorem}[section]
\newtheorem*{theorem*}{Theorem}
\newtheorem{question}[theorem]{Question}
\newtheorem{lemma}[theorem]{Lemma}
\newtheorem{proposition}[theorem]{Proposition}
\newtheorem*{proposition*}{Proposition}

\newtheorem*{corollary*}{Corollary}

\theoremstyle{examplestyle}
\newtheorem{example}[theorem]{Example}
\newtheorem{definition}[theorem]{Definition}
\newtheorem{definition*}{Definition}

\newtheorem{remark*}{Remark}

\theoremstyle{algorithmstyle}

\newcommand{\s}{{\rm S}}

\newcommand{\M}{{\rm M}}
\newcommand{\N}{{\rm N}}

\newcommand{\p}{{\rm P}}

\newcommand{\K}{{\rm K}}

\newcommand{\CC }{\mathbb{C}}

\newcommand{\suchthat}{\;\ifnum\currentgrouptype=16 \middle\fi|\;}

\begin{document}

\parindent0cm

\title[Decoupled molecules with binding polynomials of bidegree $(n,2)$]{Decoupled molecules with \\ binding polynomials of bidegree $(n,2)$}
\author{Yue Ren}
\address{Max-Planck-Institut f\"ur Mathematik in den Naturwissenschaften\\
  Inselstraße 22\\
  04103 Leipzig\\
  Germany.
}
\email{yueren@mis.mpg.de}
\urladdr{http://personal-homepages.mis.mpg.de/yueren}
\author{Johannes W. R. Martini}
\address{KWS SAAT SE\footnote{this research project is not associated with KWS SAAT SE}\\
  Grimsehlstraße 31\\
  37574 Einbeck\\
Germany.
}
\email{johannes.martini@kws.com}
\author{Jacinta Torres}
\address{Max-Planck-Institut f\"ur Mathematik in den Naturwissenschaften\\
  Inselstraße 22\\
  04103 Leipzig\\
  Germany.
}
\email{jtorres@mis.mpg.de}
\urladdr{http://personal-homepages.mis.mpg.de/jtorres}

\date{\today}

\begin{abstract}
We present a result on the number of decoupled molecules for systems binding two different types of ligands. In the case of $n$ and $2$ binding sites respectively, we show that, generically, there are $2(n!)^{2}$ decoupled molecules with the same binding polynomial. For molecules with more binding sites for the second ligand, we provide computational results.
\end{abstract}

\maketitle

\section{Introduction}
In biology, a ligand is a substance that binds to a target molecule to serve a given purpose. A classical \cite{Bohr04,hasselbalch} and intensively studied \cite{Barcroft13,Hill13} example is oxygen, which binds reversibly to hemoglobin to be transported through the bloodstream. Reversible mutual binding of different molecules is also a key feature in biological signal transduction \cite{Changeux05,Cho96,Gutierrez09,ha16} and gene regulation \cite{Gutierrez12}.

A common model for describing equilibrium and steady states of a ligand $L$ binding to the sites of a target molecule $M$ comes from the grand canonical ensemble of statistical mechanics \cite{Ben01,Hill85,Schellman75,Wyman90}. The grand partition function, in our context also known as the \emph{binding polynomial}, arises as the denominator of the rational function describing the average number of occupied binding sites as a function of ligand activity. In the case of a target molecule with only one binding site, this rational function is given by
\[\Psi(\Uplambda) = \frac{a \Uplambda}{a\Uplambda + 1},\]
where the variable $\Uplambda$ denotes the activity of the ligand in the environment, and $a$ is a transformation of the binding energy depending on the temperature, which is usually assumed to be constant. This equation is also known as the (sigmoid) Henderson-Hasselbalch titration curve. Titration refers to the laboratory method used to obtain this curve. For systems of molecules with $n$ binding sites it generalizes to the Adair equation \cite{adair25,Stefan13}:
\[\Psi(\Uplambda) = \frac{ n a_n \Uplambda^n + (n-1)a_{n-1}  \Uplambda^{n-1}+ ... + a_1 }{ a_n \Uplambda^n + a_{n-1}  \Uplambda^{n-1}+ ... + a_1+ 1}.\]

In this model, the roots of the binding polynomial play an important role for the characterization of the binding behavior of the ligand to the target molecule \cite{Briggs83,Briggs84,Briggs85,Connelly86}. The rational functions of systems with $n$ binding sites can be represented as sums of $n$ Henderson-Hasselbalch curves \cite{Martinietal2013,novelview,Onufriev04}, which means that any given system of interacting binding sites can be represented by a hypothetical molecule consisting of stochastic independent binding sites \cite{Martini13meaning} possessing the same titration curve. The roots of the binding polynomial determine the binding energies of the independent pseudo-sites in this so-called \emph{decoupled sites representation}.

For two different types of ligands, the binding polynomial has two variables representing the activities of both ligands in the environment. Seeking an analogous decoupled sites representation leads to a version in which the binding sites for the same type of ligand do not interact, but interaction terms between the sites of different ligands remain \cite{Martini13II,Martini13I}. Contrary to the case of one type of ligand, where the decoupled sites representation is unique up to permutation of the roots, there are several different decoupled molecules. It has been shown previously that in the case of $n$ and $1$ binding sites for the two ligands, respectively, there are $n!$ decoupled molecules. The situation becomes more complicated for general systems of $n_1$ and $n_2$ binding sites. The main goal of this paper is to prove the following theorem.

\begin{theorem}
The decoupled molecules with $(n,2)$ sites with a fixed binding polynomial of bidegree $(n,2)$ are the solutions to a system of $3n+2$ unknowns: the $n+2$ binding energies and the $2n$ interaction energies. Generically, the number of complex solutions to this system equals $4(n!)^3$. These come in $2(n!)^{2}$ classes under relabeling of the sites.
 \end{theorem}

The article is structured as follows: In Section~\ref{sec:backgroundAndFramework}, we recall the definition of the binding polynomial and formulate the central question addressed in this work. In Section~\ref{sec:numericalAlgebraicGeometry}, we recall some results and techniques of numerical algebraic geometry, which are necessary to prove the main theorem in Section~\ref{sec:genericDecouplesMolecules}. We conclude the article with some experimental results in Section~\ref{sec:furtherExperimentalResults} and some open questions in Section~\ref{sec:summaryAndOpenQuestion}.

\section*{Acknowledgements}
The authors would like to thank Bernd Sturmfels for suggesting this collaboration and for his comments on a previous version of this paper, Corey Harris for his useful remarks and Jon Hauenstein for his advice regarding the \textsc{bertini} computations.

\section{Background and framework}\label{sec:backgroundAndFramework}

In this section, we briefly recap the algebraic framework as well as past results, and, in doing so, fix various notations. Most importantly, we introduce some shorthand notation for molecules with $(n,2)$ sites for Sections \ref{sec:numericalAlgebraicGeometry} and \ref{sec:genericDecouplesMolecules}.

\subsection{Single type of ligand}
The binding behaviour of systems with one type of ligand is governed by the energies required to bind to each site of the target molecule and the way different binding sites interact with each other. Following the notation of \cite{Martinietal2013}, we identify target molecules with these parameters.

\begin{definition}
  A \emph{molecule} $\M$ with $n$ sites for one type of ligand is a point
  \[ \M = (g_{1}, \cdots, g_{n}, w_{1,2}, w_{1,3}, \cdots, w_{n-1,n}) \in (\mathbb C^*)^n\times (\mathbb{C}^{*})^{\binom n 2}. \]
  The $g_{i}$ are called the \textit{binding energies} and the $w_{i,j}$ are called the \textit{interaction energies}; they measure, respectively, the energy at each site $i$ and the interaction energy between sites $i$ and $j$ (see Figure~\ref{fig:moleculeMonofunctional}). We call $M$ \emph{decoupled} if $w_{i,j}=1$ for all $1\leq i < j\leq n$.

 We will consider the natural $S_n$ action that corresponds to relabelling the sites:
 \[\sigma\cdot (g_{1}, \cdots, g_{n}, w_{1,2}, \cdots, w_{n-1,n}):=(g_{\sigma(1)}, \cdots, g_{\sigma(n)}, w_{\sigma(1),\sigma(2)}, \cdots , w_{\sigma(n-1),\sigma(n)})\]
 for $\sigma\in S_n$.
\end{definition}

\begin{figure}[ht]
  \centering
  \begin{tikzpicture}[scale=0.9]
    \pgfmathsetseed{13371337}
    \coordinate (o) at (0,0);
    \draw[rounded corners=.5mm] (o) \irregularcircle{1.75cm}{1mm};
    \draw (25:1.75cm) node[draw,fill=red!30] (o2) {};
    \draw (105:1.75cm) node[draw,fill=red!30] (c) {};
    \draw (195:1.75cm) node[draw,fill=red!30] (a) {};
    \draw (285:1.75cm) node[draw,fill=red!30] {};

    \node[anchor=west,xshift=0.5cm,font=\scriptsize] (o2text) at (o2) {red sites labelled $1,\ldots,4$ for $\text{O}_2$};
    \draw[->,densely dotted] ($(o2text.west)+(0.1cm,0)$) -- (o2);
    \node[anchor=east,xshift=-0.5cm,font=\footnotesize] (aText) at (a) {binding energies $g_i$};
    \draw[->,densely dotted] ($(aText.east)+(-0.15,0)$) -- (a);
    \draw[<->,densely dotted] (c) to[out=270,in=45] node[anchor=north west,xshift=-0.35cm,yshift=0.1cm,font=\footnotesize,text width=2.25cm,text centered] {interaction\\ energies $w_{i,j}$} (a);
  \end{tikzpicture}\vspace{-0.5cm}
  \caption{Hemoglobin with its 4 sites for oxygen.}
  \label{fig:moleculeMonofunctional}
\end{figure}
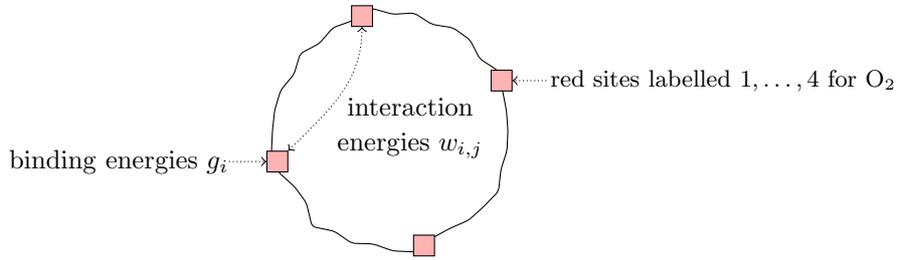

\begin{definition}
  Given a labelled target molecule $M$ with $n$ sites, we refer to $\K := \left\{0,1\right\}^{n}$ as the set of all microstates. To each microstate $k=(k_1,\ldots,k_n) \in \K$ we associated a microstate constant
  \[ g(k) := \prod_{i=1}^{n} \left(g_{i}^{k_{i}}\prod_{j=i+1}^{n}w_{i,j}^{{k}_{i}{k}_{j}}\right). \]
  The \emph{binding polynomial} is then defined as
  \[ \p_{M}(\Uplambda) = \sum_{k \in \K} g(k) \Uplambda^{|k|} \in \CC[\Uplambda]. \]
  It is a polynomial of degree $n$ with constant term $1$, and the map $M\mapsto \p_M$ is constant on the $S_n$ orbits, i.e. $\p_{\M}(\Uplambda) = \p_{\sigma(\M)}(\Uplambda)$ for every $\M \in (\mathbb{C}^{*})^{\frac{n(n+1)}{2}}$ and all $\sigma \in \s_{n}$.
\end{definition}

The following theorem is also known as the \textit{decoupled sites representation}. It implies that any molecule with real binding and interaction energies can be uniquely represented by a molecule with neutral interaction energy, provided that complex binding energies are allowed. Its proof consists of a reformulation of Vieta's formulas.

\begin{theorem}[{\cite[Proposition 2]{Martinietal2013}}]
  For any molecule $\N$ there exists a decoupled labelled target molecule $\M$, unique up to relabelling of the sites, such that $\p_{\M} = \p_{\N}$.
\end{theorem}

\subsection{Multiple types of ligands}
In case of $d>1$ types of ligands, we consider each binding site to be only able to take up to one type of ligand \cite{Martini13I}. This is sensible, as we can model a single binding site capable of binding to two types of ligands as two binding sites with interaction energies set so that the two sites can never be saturated at the same time.

For our purposes, let us assume that $d = 2$. We write $n_{1}$ and $n_{2}$ for the number of sites capable of binding to the first and second ligand, respectively.

\begin{definition}
  A \emph{molecule} $\M$ with $(n_1,n_2)$ sites is a point
  \[ \M = (g_{T_1}, \ldots, g_{T_{n_1}}, g_{S_1}, \ldots, g_{S_{n_2}}, (w_P)_{P\subset \{T_i,S_j\}, |P|=2}) \in (\mathbb C^*)^{n_1+n_2}\times (\mathbb{C}^{*})^{\binom {n_1+n_2} 2}, \]
  where $T_1,\ldots,T_{n_1},S_{1},\ldots,S_{n_2}$ represent the binding sites for ligand type $T$ and $S$ respectively (see Figure~\ref{fig:moleculeBifunctional}) and
  \begin{itemize}[leftmargin=*]
  \item $g_{T_1},\ldots,g_{T_{n_1}}$ and $g_{S_1},\ldots,g_{S_{n_2}}$ are the binding energies,
  \item $w_P$ for $P\subset\{T_1,\ldots,T_{n_1},S_1,\ldots,S_{n_2}\}$ with $|P|=2$ are the interaction energies.
  \end{itemize}
  We call $\M$ \emph{decoupled}, if $w_P=1$ for $P\subset \{T_1,\ldots,T_{n_1}\}$ and $P\subset\{S_1,\ldots,S_{n_2}\}$.

  Similar to the case $d=1$, there is a natural $S_{n_1}\times S_{n_2}$ action that corresponds to relabelling the sites.
\end{definition}

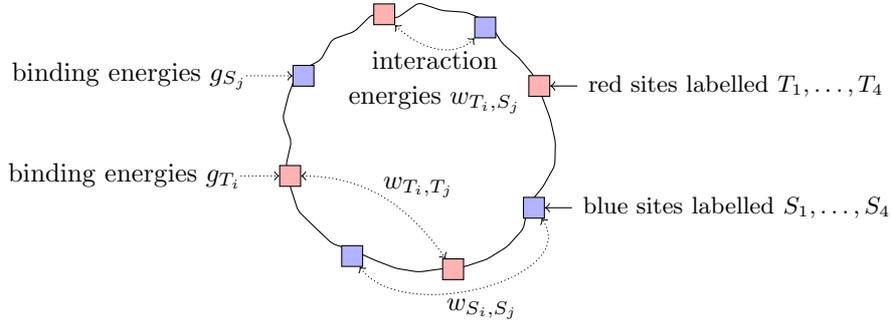
\begin{figure}[ht]
  \centering
  \begin{tikzpicture}
    \pgfmathsetseed{13379773}
    \coordinate (o) at (0,0);
    \draw[rounded corners=.5mm] (o) \irregularcircle{1.75cm}{1mm};
    \draw (25:1.75cm) node[draw,fill=red!30] (o2) {};
    \draw (60:1.75cm) node[draw,fill=blue!30] (b) {};
    \draw (105:1.75cm) node[draw,fill=red!30] (c) {};
    \draw (150:1.75cm) node[draw,fill=blue!30] (f) {};
    \draw (195:1.75cm) node[draw,fill=red!30] (a) {};
    \draw (240:1.75cm) node[draw,fill=blue!30] (d) {};
    \draw (285:1.75cm) node[draw,fill=red!30] (e) {};
    \draw (330:1.75cm) node[draw,fill=blue!30] (co2) {};

    \node[anchor=west,xshift=0.5cm,font=\scriptsize] (o2text) at (o2) {red sites labelled $T_1,\ldots,T_4$};
    \node[anchor=west,xshift=0.5cm,font=\scriptsize] (co2text) at (co2) {blue sites labelled $S_1,\ldots,S_4$};
    \draw[->] (o2text) -- (o2);
    \draw[->] (co2text) -- (co2);
    \node[anchor=east,xshift=-0.5cm,font=\footnotesize] (aText) at (a) {binding energies $g_{T_i}$};
    \draw[->,densely dotted] ($(aText.east)+(-0.15,0)$) -- (a);
    \node[anchor=east,xshift=-0.6cm,font=\footnotesize] (fText) at (f) {binding energies $g_{S_j}$};
    \draw[->,densely dotted] ($(fText.east)+(-0.15,0)$) -- (f);
    \draw[<->,densely dotted] (b) to[out=225,in=315] node[anchor=north,yshift=0.1cm,font=\footnotesize,text width=2.5cm,text centered] {interaction\\ energies $w_{T_i,S_j}$} (c);
    \draw[<->,densely dotted] (a) to[out=0,in=125] node[anchor=south west,xshift=-0.15cm,yshift=-0.15cm,font=\footnotesize] {$w_{T_i,T_j}$} (e);
    \draw[<->,densely dotted] (d) to[out=305,in=305] node[anchor=north,yshift=0cm,font=\footnotesize] {$w_{S_i,S_j}$} (co2);
  \end{tikzpicture}\vspace{-0.5cm}
  \caption{A molecule with (4,4) sites.}
  \label{fig:moleculeBifunctional}
\end{figure}

\begin{definition}
  Similarly to the case $d=1$, we can define the binding polynomial $\p_{\M}$ of a molecule $\M$. Explicitly, for decoupled molecules $\M$, $\p_{\M}$ is a bivariate polynomial in the two (ligand) variables $\Uplambda_{1}$ and $\Uplambda_{2}$,
  \[ \p_{\M}(\Uplambda_{1},\Uplambda_{2})  = \sum_{i=1,\ldots,{n_1}} \sum_{j=1,\ldots,n_2} a_{i,j} \Uplambda_{1}^i\Uplambda_{2}^j, \]
  where the coefficients $a_{i,j}$ are given by
  \begin{equation}
    \label{oursystem}
    a_{i,j} = \underset{|I| = i, \, |J|= j}{\underset{J \subset \{1, \cdots, n_{2}\}}{\underset{I \subset \{1, \cdots, n_{1}\}}{\sum}}} \;\prod_{T_i\in I} g_{T_i} \prod_{S_j\in J} g_{S_j} \underset{T_i \in J}{\underset{S_j \in I}{\prod}} w_{\{T_i,S_j\}}.
  \end{equation}
  It is a bivariate polynomial of bidegree $(n_1,n_2)$ with constant term $1$. Moreover, the map $\M\mapsto \p_\M$ is constant on the $S_{n_1}\times S_{n_2}$-orbits, i.e. $\p_{\M}(\Uplambda) = \p_{\sigma(\M)}(\Uplambda)$ for every $\M \in (\mathbb C^*)^{n_1+n_2}\times (\mathbb{C}^{*})^{\binom {n_1+n_2} 2}$ and all $\sigma \in \s_{n_1}\times\s_{n_2}$.
\end{definition}

In this case, the decoupled sites representation takes the following form.

\begin{theorem}[{\cite[Corollary 2]{Martini13I}}]
  For any molecule $\N$ with $(n,1)$ sites there exists, up to relabelling of the sites, and counted with multiplicity, $n!$ decoupled molecules $\M$ of the same type such that $\p_\N = \p_\M$.
\end{theorem}

\subsection{Decoupled molecules with $(n,2)$ sites}
\label{2.3}
The main focus of this article are decoupled molecules with $(n,2)$ sites, for which we will simplify the notation as follows: instead of $T_1,\ldots,T_n$, we label the $n$ binding sites of the first type with $1,\ldots,n$, and, instead of $S_1,S_2$, we label the two binding sites of the second type with $A,B$ (see Figure~\ref{fig:molecule(n,2)}), so that
\begin{itemize}[leftmargin=*]
\item $g_1,\ldots,g_n,g_A,g_B$ represent the binding energies,
\item $w_{1,A},\ldots,w_{n,A},w_{1,B},\ldots,w_{n,B}$ represent the non-trivial interaction energies.
\end{itemize}
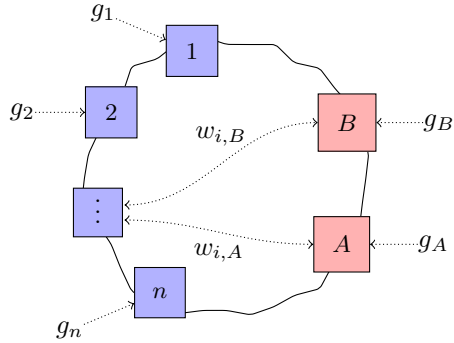
\begin{figure}[ht]
  \centering
  \begin{tikzpicture}[every node/.style={font=\scriptsize},square/.style={regular polygon,regular polygon sides=4}]
    \pgfmathsetseed{13379999}
    \coordinate (o) at (0,0);
    \draw[rounded corners=.5mm] (o) \irregularcircle{1.85cm}{1mm};
    \draw (330:1.75cm) node[square,draw,fill=red!30] (a) {$A$};
    \draw (25:1.75cm) node[square,draw,fill=red!30] (b) {$B$};
    \draw (105:1.75cm) node[square,draw,fill=blue!30] (1) {$1$};
    \draw (150:1.75cm) node[square,draw,fill=blue!30] (2) {$2$};
    \draw (195:1.75cm) node[square,draw,fill=blue!30,inner sep=2.5pt] (i) {$\rvdots$};
    \draw (240:1.75cm) node[square,draw,fill=blue!30] (n) {$n$};

    \node[anchor=west,xshift=1cm,font=\footnotesize,inner sep=0pt] (atext) at (a) {$g_A$};
    \node[anchor=west,xshift=1cm,font=\footnotesize,inner sep=0pt] (btext) at (b) {$g_B$};
    \draw[->,densely dotted] (atext) -- (a);
    \draw[->,densely dotted] (btext) -- (b);

    \node[anchor=east,xshift=-1cm,yshift=0.5cm,inner sep=0pt,font=\footnotesize] (1text) at (1) {$g_1$};
    \node[anchor=east,xshift=-1cm,inner sep=0pt,font=\footnotesize] (2text) at (2) {$g_2$};
    \node[anchor=east,xshift=-1cm,yshift=-0.5cm,inner sep=0pt,font=\footnotesize] (ntext) at (n) {$g_n$};
    \draw[->,densely dotted] (1text) -- (1);
    \draw[->,densely dotted] (2text) -- (2);
    \draw[->,densely dotted] (ntext) -- (n);

    \draw[<->,densely dotted] ($(i)+(0.35,-0.1)$) to[out=0,in=180] node[anchor=north,font=\footnotesize] {$w_{i,A}$} (a);
    \draw[<->,densely dotted] ($(i)+(0.35,0.1)$) to[out=0,in=180] node[anchor=south,yshift=0.1cm,font=\footnotesize] {$w_{i,B}$} (b);
  \end{tikzpicture}\vspace{-0.25cm}
  \caption{A decoupled molecule with (n,2) sites.}
  \label{fig:molecule(n,2)}
\end{figure}

The formulas for the coefficients of the binding polynomial then simplify to the polynomials in System (\ref{thesystem}).
For an explicit instance of the equations and their solutions, see Section~\ref{sec:explicitSolutions}.
\begin{figure}[ht]
  \centering
  \begin{equation}
  \label{thesystem}
  \begin{aligned}
    a_{1,0} &= g_1+\ldots+g_n, \\[-0.2cm]
    &\;\,\vdots \\[-0.2cm]
    a_{n,0} &= g_1\cdot \ldots\cdot g_n, \\
    a_{0,1} &= g_A + g_B,\\
    a_{1,1} &= g_A (g_1w_{1,A}+\ldots+g_nw_{n,A})+ g_B (g_1w_{1,B}+\ldots+g_nw_{n,B}),\\
    a_{2,1} &= g_A (g_1g_2w_{1,A}w_{2,A}+\ldots+g_{n-1}g_nw_{n-1,A}w_{n,A})\\
    &\qquad +g_B (g_1g_2w_{1,B}w_{2,B}+\ldots+g_{n-1}g_nw_{n-1,B}w_{n,B}), \\[-0.5cm]
    &\;\,\vdots \\[-0.2cm]
    a_{n,1} &= g_A g_1\cdots g_nw_{1,A}\cdots w_{n,A} + g_B g_1\cdots g_nw_{1,B}\cdots w_{n,B},\\
    a_{0,2} &= g_Ag_B,\\
    a_{1,2} &= g_Ag_B (g_1 w_{1,A}w_{1,B}+\ldots+g_n w_{n,A}w_{n,B}),\\
    a_{2,2} &= g_Ag_B (g_1 w_{1,A}w_{1,B}g_2 w_{2,A}w_{2,B}+\ldots+g_{n-1} w_{n-1,A}w_{n-1,B}g_n w_{n,A}w_{n,B}), \\[-0.2cm]
    &\;\,\vdots \\[-0.2cm]
    a_{n,2} &= g_Ag_B g_1 w_{1,A}w_{1,B}\cdots g_n w_{n,A}w_{n,B}.
  \end{aligned}
\end{equation}\vspace{-0.5cm}
\caption{Coefficients of the binding polynomial of bidegree (n,2).}
\label{fig:binding2}
\end{figure}
We denote the pair set of decoupled molecules with $(n,2)$ sites and their binding polynomials of bidegree $(n,2)$ (ignoring its constant term $1$) by
\[ \mathcal M = \left\{ (\underline g,\underline w;\underline a)\in (\mathbb C^*)^{n+2}\times (\mathbb{C}^{*})^{n\cdot 2} \times \mathbb C^{(n+1)\cdot (2+1)-1}\suchthat (\underline g,\underline w;\underline a) \text{ satisfies System (\ref{thesystem})}\right\}.\]

\pagebreak
\section{Numerical algebraic geometry}\label{sec:numericalAlgebraicGeometry}

In this section we recall some basic notions of numerical algebraic geometry and its main workhorse: homotopy continuation. For that, we regard $\mathcal M$ as the kernel of the polynomial map
\[ f: \underbrace{(\mathbb C^*)^{n+2}\times (\mathbb{C}^{*})^{n\cdot 2}}_{=:X} \times \underbrace{\mathbb C^{(n+1)\cdot (2+1)-1}}_{=:Y} \rightarrow \mathbb C^{(n+1)\cdot (2+1)-1} \]
with
\[
  f(\underline g,\underline w; \underline a) =
  \begin{pmatrix}
     g_1+\ldots+g_n - a_{1,0} \\
     \vdots \\[-0.25cm]
     \vdots
  \end{pmatrix},
\]
where $\underline g:=(g_1,\ldots,g_n,g_A,g_B)$ and $\underline w:=(w_{1,A},\ldots,w_{n,A},w_{1,B},\ldots,w_{n,B})$ are referred to as \emph{unkowns}, and $\underline a=(a_{1,0},\ldots,a_{n,2})$ are regarded as \emph{parameters}. We fix a projection
\[ \pi_Y : X\times Y\longrightarrow Y. \]

One fundamental and important concept is that solutions vary continuously in the parameters, which is summarized in the following theorem.

\begin{theorem}[{\cite[Theorem A.14.1]{SW05}}]\label{thm:continuity}
  If there is an isolated solution $(\underline g^*,\underline w^*;\underline a^*)\in X\times Y$ of $f(\underline g,\underline w;\underline a^*)=0$, then there are arbitrarily small euclidean open sets $\mathcal U\subset X$ that contain $(\underline g^*,\underline w^*)$ and such that
  \begin{enumerate}[leftmargin=*]
  \item $(\underline g^*,\underline w^*;\underline a^*)$ is the only solution of $f(\underline g,\underline w;\underline a^*)=0$ in $\mathcal U\cap (X\times \{\underline a^*\})$;
  \item $f(\underline g,\underline w;\underline a')=0$ has only isolated solutions for $a'\in \pi_{Y}(\mathcal U)$ and $(\underline g,\underline w)\in \mathcal U\cap (X\times \{\underline a^*\})$;
  \item the multiplicity of $(\underline g^*,\underline w^*;\underline a^*)$ as a solution of $f(\underline g,\underline w;\underline a^*)=0$ equals the sum of the multiplicities of the isolated solutions of $f(\underline g,\underline w;\underline a')=0$ for $a'\in \pi_{Y}(\mathcal U)$ and $(\underline g,\underline w)\in \mathcal U\cap (X\times \{\underline a^*\})$
  \end{enumerate}
\end{theorem}

\begin{example}[Vieta's Formula]\label{ex:vieta}
  Consider the first $n$ components of our polynomial map, which are given by (abbreviating $a_i:=a_{i,0}$):
  \[
    f(g_1,\ldots,g_n;a_1,\ldots,a_n) :=
    \begin{pmatrix}
      g_1+\ldots+g_n-a_1 \\
      g_1g_2 + g_1g_3 + \ldots + g_{n-1}g_n-a_2 \\
      \vdots \\
      g_1\cdot \ldots \cdot g_n - a_n
    \end{pmatrix}.
  \]
  Given any parameter $\underline a'\in \CC^n$, Vieta's formula states that any solution $\underline g'\in \CC^n$ to $f(\underline g;\underline a')=0$ consists of the roots of the univariate polynomial $t^n+a_1' t^{n-1} + \ldots + a_n'$. Hence there exist a Zariski-open set $\mathcal U:= \CC^n \setminus\text{Disc}_x(x^n+a_1 x^{n-1} + \ldots + a_n)$ such that for any $\underline a'\in \mathcal U$ there are $n!$ distinct simple solutions to $f(\underline g;\underline a')=0$. We say that there are \emph{generically} $n!$ solutions and refer to $\underline a'\in \mathcal U$ as a \emph{generic} choice of parameters.

  Should $x^n+a_1' x^{n-1} + \ldots + a_n'=(x-1)^n$, then the only solution is $\underline g'=(1,\ldots,1)$. Theorem \ref{thm:continuity} implies that this solution is of multiplicity $n!$. This will be important in the proof of Lemma \ref{lem:normalized}.
\end{example}

The arguably most essential tool in numerical algebraic geometry is path tracking. That is given
\begin{itemize}[leftmargin=*]
\item a starting solution $(g',w';a')\in X\times Y$
\item a target parameter $a^*\in Y$
\item a continuous path $\phi:[0,1]\rightarrow Y$ with $\phi(1)=a'$ and $\phi(0)=a^\ast$
\end{itemize}
there exist, under certain circumstances \cite[Theorem 7.1.6]{SW05}, a \emph{solution path}
\[ z: (0,1] \rightarrow X \text{ with } z(1)=(g',w') \text{ and } f(z(t),\phi(t))=0.\]
However, the solution path might diverge, which is why these problems are commonly studied in a projective framework.

\begin{example}[solutions at infinity]
  The simplest example of diverging solution path is the function
  \[ f: \CC \times \CC \longrightarrow \CC,\quad f(x;a)=ax^2-x, \]
  with two starting solutions $(0;1)$, $(1;1)$, the target parameter $0$ and the continuous, straight-line path $\phi: [0,1] \rightarrow \CC, t\mapsto 1-t$, see Figure~\ref{fig:solutionsAtInfinity}.
  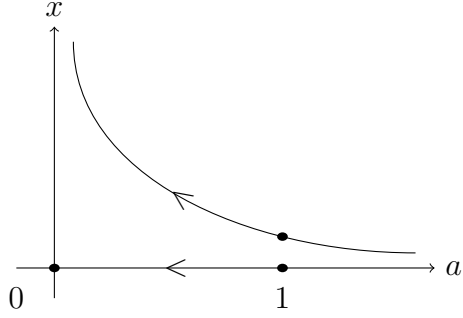
\begin{figure}[ht]
    \centering
    \begin{tikzpicture}[yscale=0.8]
        \draw[->] (-0.5, 0) -- (5, 0) node[right] {$a$};
        \draw[->] (0, -0.5) -- (0, 4) node[above] {$x$};
        \node at (1.6,0) {$<$};
        \draw (4.75,0.25) to[out=180,in=270] node[sloped] {$<$} (0.25,3.75);
        \fill (3,0) circle (2pt);
       \node at (3,-0.5) {$1$};
        \fill (0,0) circle (2pt);
       \node at (-0.5,-0.5) {$0$};
        \fill (3,0.52) circle (2pt);
    \end{tikzpicture}\vspace{-0.5cm}
    \caption{A converging and a diverging solution path.}
    \label{fig:solutionsAtInfinity}
  \end{figure}

  The two solution paths are
  \begin{align*}
    z_1: & \;\; [0,1) \longrightarrow \CC, \quad t\longmapsto 0, \\
    z_2: & \;\; [0,1) \longrightarrow \CC, \quad t\longmapsto \frac{1}{1 - t  },
  \end{align*}
  of which the first obviously converges, while the second diverges. Note that diverging paths can only appear if parameters occur in the coefficients of non-constant monomials, which is not the case in System~(\ref{thesystem}), see proof of Theorem~\ref{thm:generic}.
\end{example}

\section{Generic decoupled molecules with $(n,2)$ sites}\label{sec:genericDecouplesMolecules}

In this section, we show that a binding polynomial represents generically $4\cdot (n!)^3$ decoupled molecules with $(n,2)$ sites. Due to the complexity of the system of polynomial equations, the proof is split in two parts. First, we study a special class of decoupled molecules and their binding polynomials. In a second step, we study their implication to the generic case.

\subsection{Normalized molecules}

In this subsection, we restrict ourselves to a special class of decoupled molecules and their binding polynomials. This simplifies our system of equations and allows us to show that their binding polynomials generically represent $2n!$ molecules, each of multiplicity $2(n!)^{2}$.

\begin{definition}
  Recall that $\mathcal M$ consists pairs of molecules and their binding polynomials  (see Section~\ref{2.3}). We define the set of all \emph{normalized} molecules to be
  \[ \mathcal M_{\text{norm}}:=\left\{ (\underline g,\underline w;\underline a)\in\mathcal M\;\middle|
      \begin{array}{c}
        g_i=1 \text{ for } i=1,\ldots,n \text{ and } g_A=g_B=1 \\
        w_{i,A} w_{i,B}=1 \text{ for } i=1,\ldots,n
      \end{array}
    \right\}. \]
\end{definition}

\begin{lemma}\label{lem:normalized}
  The projection
  \[ \mathcal M\stackrel{\pi_{\text{norm}}}\longrightarrow (\CC^\ast)^n\times \CC^n, \quad (\underline g,\underline w;\underline a)\longmapsto (w_{1,A},\ldots,w_{n,A};a_{1,1},\ldots,a_{n,1})\]
  maps $\mathcal M_{\text{norm}}$ bijectively onto the affine variety $V$ cut out by
  \begin{equation}\label{eqs:normalized}
    \begin{aligned}
      a_{1,1} &= (w_{1,A}+\ldots+w_{n,A})+ \left(\frac{1}{w_{1,A}}+\ldots+\frac{1}{w_{n,A}}\right),\\[-0.25cm]
      &\;\,\vdots \\[-0.25cm]
      a_{n,1} &= w_{1,A}\cdots w_{n,A} + \frac{1}{w_{1,A}\cdots w_{n,A}}.
    \end{aligned}
  \end{equation}
  Moreover, any point on $V$ of multiplicity $1$ is the image of a point on $\mathcal M_{\text{norm}}$ of multiplicity $2(n!)^2$.
\end{lemma}
\begin{proof}
  The bijection follows directly from the conditions on $\mathcal M_{\text{norm}}$ and the equations of System (\ref{thesystem}): If $(\underline g,\underline w;\underline a)$ is normalized, then by definition $\underline g=(1,\ldots,1)$ and $w_{1,B} = w_{1,A}^{-1}$. Additionally, the following parameters are uniquely determined by the following equations of System (\ref{thesystem}):
  \begin{align*}
    a_{0,1} &= g_A + g_B,      &          a_{0,2} &= g_Ag_B, \\
    a_{1,0} &= g_1+\ldots+g_n, &          a_{1,2} &= g_Ag_B (g_1 w_{1,A}w_{1,B}+\ldots+g_n w_{n,A}w_{n,B}),\\
            &\;\,\vdots        &                  &\;\,\vdots \\
    a_{n,0} &= g_1\cdot \ldots\cdot g_n,& a_{n,2} &= g_Ag_B g_1 w_{1,A}w_{1,B}\cdots g_n w_{n,A}w_{n,B}.
  \end{align*}
  The multiplicity follows from the fact that:
  \begin{itemize}[leftmargin=*]
  \item any solution $(g_A,g_B;a_{0,j})$ to the two equations in the first row is of multiplicity $2$ (see Example~\ref{ex:vieta}),
  \item the solution $(g_1,\ldots,g_n;a_{i,0})=(1,\ldots,1;\binom{n}{i})$ to the latter equations in the first column is of multiplicity $n!$,
  \item given $g_A=g_B=g_i=1$, the solution $(w_{i,A},w_{i,B};a_{i,2})=(1,\ldots,1;\binom{n}{i})$ to the latter equations in the second column is of multiplicity $n!$.
  \end{itemize}
  and from the fact that the multiplicity of the entire system equals the product of the multiplicities of the three smaller systems in our case \cite[Proposition~1.29]{EisenbudHarris16}.
\end{proof}

\begin{proposition}\label{prop:degenerate}
  A normalized binding polynomial represents generically $2n!$ decoupled molecules, each of multiplicity $2(n!)^{2}$.
\end{proposition}
\begin{proof}
  By Lemma~\ref{lem:normalized}, it suffices to show that System (\ref{eqs:normalized}) has $2n!$ simple solutions for generic $\underline a = (a_{1,1},\ldots,a_{n,1})\in \CC^n$, or rather for $(a_{1,1},\ldots,a_{n,1})\in\mathcal U$ for some Euclidean open subset $\mathcal U\subseteq \CC^n$. For the sake of simplicity, we abbreviate $a_i:= a_{i,1}$ and $w_i:=w_{i,A}$ for $i=1,\ldots,n$.
  Next, we introduce $n$ new variables $\mu_1,\ldots,\mu_n$ and consider the following equivalent system of $2n$ equations in the $2n$ variables $\mu_1,\ldots,\mu_n,w_1\ldots,w_n$:
  \begin{align}
    \tag{1}\mu_1 &= (w_1+\ldots+w_n)\\[-0.25cm]
    \tag*{\vdots\;\;}&\;\;\vdots \\[-0.25cm]
    \tag{n}\mu_n &= w_1\cdots w_n\\
    \tag{-1}a_1 -\mu_1&= \left(\frac{1}{w_1}+\ldots+\frac{1}{w_n}\right)\\[-0.25cm]
    \tag*{\vdots\;\;}&\;\;\vdots \\[-0.25cm]
    \tag{-n}a_n -\mu_n&= \frac{1}{w_1\cdots w_n}
  \end{align}

  Let $\mathcal N$ be the variety cut out by the system above and let $\pi_\mu$ and $\pi_a$ denote the three projections onto  $\mu_i$ and $a_i$ respectively.

  \begin{center}
    \begin{tikzpicture}
      \node[anchor=base west] (N) at (0,0) {$\mathcal N$};
      \node[anchor=base west] (subseteq) at (N.base east) {$\subseteq$};
      \node[anchor=base west] (Cw) at (subseteq.base east) {$(\CC^\ast)^{|\{w_1,\ldots,w_n\}|}$};
      \node[anchor=base west] (times1) at (Cw.base east) {$\times$};
      \node[anchor=base west] (Ca) at (times1.base east) {$\CC^{|\{a_1,\ldots,a_n\}|}$};
      \node[anchor=base west] (times2) at (Ca.base east) {$\times$};
      \node[anchor=base west] (Cm) at (times2.base east) {$\CC^{|\{\mu_1,\ldots,\mu_n\}|}$};
      \node[anchor=base west,yshift=-1.5cm] (CaLow) at (Ca.base west) {$\CC^{|\{a_1,\ldots,a_n\}|}$};
      \node[anchor=base west,xshift=1cm,yshift=-1.5cm] (CmLow) at (Cm.base west) {$\CC^{|\{\mu_1,\ldots,\mu_n\}|}$};
      \draw[->] (Ca) -- node[right] {$\pi_a$} (CaLow);
      \draw[->] (Cm) -- node[right] {$\pi_\mu$} (CmLow);
      \draw[->,dashed] (CmLow) -- node[above] {$\varphi$} (CaLow);
    \end{tikzpicture}
  \end{center}
  We will construct a dominant (i.e. its image is Zariski dense), 2:1 rational map
  \[ \varphi: \CC^n \dashrightarrow \CC^n, \quad (w_1,\ldots,w_n) \longmapsto (a_1,\ldots,a_n), \]
  that maps $\underline \mu$ to the unique $\underline a$ for which some $\underline w$ exists such that $(\underline w;\underline a,\underline \mu)\in\mathcal N$. In short, the diagram above commutes. The image of the complement $\CC^n\setminus\text{Disc}_x(x^n+\mu_1 x^{n-1} + \ldots + \mu_n)$ will then contain an open set $\mathcal U\subseteq \CC^n$, and for any $\underline a\in\mathcal U$ the system will have $2n!$ solutions: $2$~solutions in $\underline \mu$, both outside the discriminant, and consequently also $n!$ solutions in $\underline w$ for each $\underline\mu$.

  To construct $\varphi$, observe that combining equations $(-n)$ and $(n)$ and obtain:
  \[ a_n -\mu_n = \frac{1}{w_1\cdots w_n} = \frac{1}{\mu_n}, \]
  which is equivalent to
  \begin{equation*}
    \mu_n^2-a_n\cdot \mu_n+1 = 0 \quad\text{or}\quad a_n=\mu_n+\frac{1}{\mu_n}.
  \end{equation*}
  Moreover, multiplying equation $(-(n-1))$ with $x_1\cdots x_n$ yields
  \[ (a_{n-1}-\mu_{n-1})\cdot \underbrace{x_1\cdots x_n}_{\overset{\text{Eq.}(n)}{=} \mu_n} = \underbrace{x_1+\ldots+x_n}_{\overset{\text{Eq.}(1)}{=} \mu_1}, \]
  or, more generally, by multiplying Equation $(-i)$ with $x_1\cdots x_n$:
  \begin{equation*}
    (a_i-\mu_i)\cdot \mu_n = \mu_{n-i} \quad\text{or}\quad a_i = \mu_i - \frac{\mu_{n-i}}{\mu_n}\quad \text{for } i=1,\ldots,n-1.
  \end{equation*}
  Set
  \begin{align*}
    \varphi: & & \CC^n & \dashrightarrow \CC^n,\\
    & & (\mu_1,\ldots,\mu_n) & \longmapsto \left(\mu_1-\frac{\mu_{n-1}}{\mu_n},\ldots,\mu_{n-1}-\frac{\mu_1}{\mu_n},\mu_n-\frac{1}{\mu_n}\right)
  \end{align*}
  By construction, $\varphi$ is $2:1$ and commutes with the projections $\pi_\mu$, $\pi_a$. Moreover, it is dominant as its Jacobian,
  \[ J(\phi) =
    \begin{pmatrix}
      1 & & & & -1 \\
        & 1 & & -1 \\
        & & \ddots \\
        & -1 & & 1 \\
        \frac{\mu_{n-1}}{\mu_n^2} & \frac{\mu_{n-2}}{\mu_n^2} & \cdots & \frac{\mu_{1}}{\mu_n^2} & 1+ \frac{1}{\mu_n^2}
      \end{pmatrix} \]
    is invertible at $(1,\ldots,1)$.
\end{proof}

\subsection{A generic decoupled sites representation}

In this subsection, we will infer from Proposition~\ref{prop:degenerate} the number of molecules a binding polynomial generically represents.

\begin{figure}[ht]
  \centering
  \begin{tikzpicture}[xscale=1.25]
    \draw[->] (-4,-0.25) -- (4,-0.25);
    \draw (0,-0.1) -- (0,-0.4);
    \node[anchor=north] at (0,-0.4) {$\phantom{{}^\ast}\underline a^\ast$};
    \node[anchor=west] at (-0.85,-0.25) {$($};
    \node[anchor=east] at (0.85,-0.25) {$)$};
    \draw (0.5,-0.15) -- (0.5,-0.35);
    \node[anchor=south] at (0.5,-0.15) {$\phantom{{}^\ast}\underline a'$};

    \draw[->] (-4,0) -- (-4,6);
    \draw[->] (-3.75,2.25) -- node[right] {$\pi_{\underline a}$} (-3.75,1);
    \draw (-4.1,4.91) -- (-3.9,4.91);
    \node[anchor=east] at (-4.1,4.91) {$(\underline g^\ast, \underline w^\ast)$};

    \node[anchor=west] at (-3.5,4.5) {$\textcolor{red}{\mathcal M_{\text{norm}}}\subseteq \mathcal M$};

    \draw plot[smooth, tension=1] coordinates { (-3.5,5.5) (-0.175,4.75) (-1,3.75) (0,3) (-1,2.25) (-0.175,1.25) (-3.5,0.5)};
    \draw plot[smooth, tension=1] coordinates { (3.5,5.5) (0.175,4.75) (1,3.75) (0,3) (1,2.25) (0.175,1.25) (3.5,0.5)};

    \draw[loosely dotted] (0,0) -- (0,6);
    \fill[red] (0,1.475) circle (0.5mm);
    \fill[red] (0,3) circle (0.5mm);
    \fill[red] (0,4.525) circle (0.5mm);
    \node[anchor=east,red] at (-1,3) {$2n!$ pts};
    \draw[white,line width=3pt,shorten <= 2pt,shorten >= 5pt] (-1.1,2.85) -- (0,1.475);
    \draw[red,->,shorten <= 2pt,shorten >= 5pt] (-1.1,2.85) -- (0,1.475);
    \draw[red,->,shorten <= 2pt,shorten >= 5pt] (-1.1,3) -- (0,3);
    \draw[white,line width=3pt,shorten <= 2pt,shorten >= 5pt] (-1.1,3.15) -- (0,4.525);
    \draw[red,->,shorten <= 2pt,shorten >= 5pt] (-1.1,3.15) -- (0,4.525);

    \draw[dashed] (0,4.535) circle (8.5mm);
    \node[anchor=west] at (0.5,5.5) {$\subseteq U$};
    \fill (0.5,4.91) circle (0.5mm);
    \fill (0.5,4.18) circle (0.5mm);
    \node[anchor=west] at (1,4.525) {$2(n!)^2$ pts each};
    \draw[->,shorten <= 2pt,shorten >= 5pt] (1.1,4.535) -- (0.5,4.91);
    \draw[->,shorten <= 2pt,shorten >= 5pt] (1.1,4.515) -- (0.5,4.18);
  \end{tikzpicture}\vspace{-0.5cm}
  \caption{Pertubation of normalized binding polynomials.}
  \label{fig:pertubationOfNormalizedBindingPolynomials}
\end{figure}
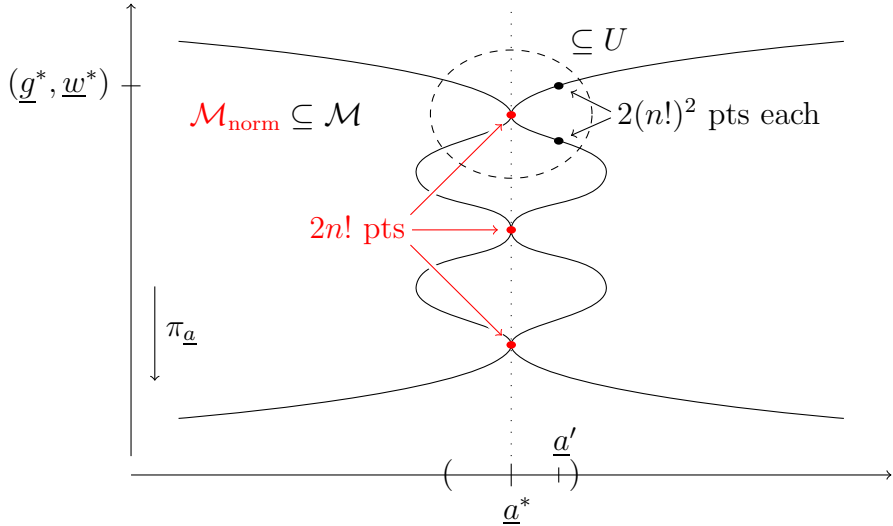

\begin{theorem}\label{thm:generic}
  A binding polynomial of bidegree $(n,2)$ represents generically $4(n!)^3$ decoupled mole\-cules with $(n,2)$ sites. These come in $2(n!)^2$ classes modulo the $S_n\times S_2$ action that corresponds to relabelling of the sites.
\end{theorem}
\begin{proof}
  It suffices to show the claim in a euclidean open set. For that consider a generic normalized binding polynomial $\underline a^\ast \in \CC^{(n+1)(2+1)-1}$. Proposition~\ref{prop:degenerate} states that there are $2n!$ solutions $(\underline g^*,\underline w^*;\underline a^*)$ to $f(\underline g,\underline w;\underline a^\ast)=0$ of multiplicity $2(n!)^2$. Applying Theorem~\ref{thm:continuity} to each of the solutions, we obtain an open subset $U\subseteq X\times Y$, where $X:=(\mathbb{C}^*)^{n+2} \times (\mathbb{C}^*)^{n\cdot 2}$, $Y:=\mathbb{C}^{(n+1)(2+1)-1}$, such that
  \begin{itemize}[leftmargin=*]
  \item for any $\underline a'\in \pi_{\underline a}(U)$, $U\cap X \times \{a'\})$ has only isolated solutions of $f(\underline g, \underline w;\underline a')=0$,
  \item the sum of the multiplicities of those isolated solutions is $4(n!)^3$.
  \end{itemize}
  It remains to show that $U$ contains all isolated solutions of $f(\underline g, \underline w;\underline a')=0$ for $\underline a'\in \pi_{\underline a}(U)$ and that the solutions are simple. Both follow from the fact that our parameters are exactly the constant terms, i.e. we can regard $\mathcal M$ as the graph of the polynomial map
  \[ h:X \longrightarrow Y, \quad (\underline g,\underline w) \longmapsto
    \begin{pmatrix}
      g_1+\ldots+g_n \\
      \vdots \\
      \vdots
    \end{pmatrix}, \]
  so that $f(\underline g,\underline w;\underline a)=0$ is equivalent to $h(\underline g,\underline w)=\underline a$. Fix $a'\in\pi_{\underline a}(U)$ and a path
  \[ \phi:[0,1]\rightarrow Y,\quad \phi(0) = a^* \text{ and } \phi(1) = a'. \]

  To see that $U$ contains all solutions to $f(\underline g, \underline w;\underline a')=0$, observe that any solution $(g',w';a')$ has a solution path
  \begin{align*}
    z:(0,1]\longrightarrow X \text{ with } z(1)=(g',w')
  \end{align*}
  such that $f(z(t);\phi(t))=0$ for all $t\in (0,1]$. As $\lim_{t\rightarrow 0} h(z(t)) = \lim_{t\rightarrow 0} \phi(t) = \underline a^\ast$ and $h$ is polynomial, $\lim_{t\rightarrow 0} z(t)$ has to converge. Therefore $(\lim_{t\rightarrow 0} z(t),\underline a^*)$ is one of our $2n!$ solutions, which implies $(g',w';a')\in U$.

  To see that solutions are generically simple note that a solution $(g',w';a')$ to \linebreak $f(\underline g,\underline w;\underline a')=0$ is singular if and only if the point $(g',w')$ is a critical point of $h$. Hence, any solution in the following open set will be simple
  \[ \mathcal U :=  U\cap\pi_{\underline a}^{-1}(Y\setminus S), \]
  where $S:= \{ h(g',w')\in \CC^{\{a_{i,j}\}} \mid (g',w') \text{ critical point}\}$ consists of the images of all critical points of $h$.
\end{proof}

\section{Further experimental results}\label{sec:furtherExperimentalResults}

In this section, we provide some experimental results for $(n,2)$ and beyond. For simplicity, we will use randomly chosen $\underline a \in\CC^{(n+1)(2+1)-1}$. Moreover, we will also fix a choice of $g_{S_1},\ldots,g_{S_{n_1}},g_{T_1},\ldots,g_{T_{n_2}}$ to factor out the natural $S_{n_1}\times S_{n_2}$ action on the roots of System (\ref{oursystem}), see Section~\ref{sec:explicitSolutions}.

All computations are done using one of the following three programs:
\begin{description}[labelwidth=5em,leftmargin =\dimexpr\labelwidth+\labelsep\relax, font=\sffamily\mdseries]
\item[\textsc{bertini}\cite{bertini}] A solver for polynomial equations using numerical algebraic geometry. It has built-in features for parallel path-tracking, which proved to be particularly useful for big examples.
\item[\textsc{gfan}\cite{gfan}] A software package for computing Gr\"obner fans and tropical varieties. It features a new algorithm for computing mixed volumes using tropical homotopy methods \cite{Jensen16}.
\item[\textsc{Singular}\cite{Singular}] A computer algebra system for polynomial computations, with special emphasis on commutative and non-commutative algebra, algebraic geometry, and singularity theory.
\end{description}
Code, tutorials and other auxiliary files for the computations will soon be made available under \texttt{software.mis.mpg.de}.

\subsection{Explicit solutions for $(3,2)$}\label{sec:explicitSolutions}

Consider the following equations of System (\ref{thesystem}) for $n=3$:
\begin{align*}
  a_{1,1} &= g_A (g_1w_{1,A}+g_2w_{2,A}+g_3w_{3,A})+ g_B (g_1w_{1,B}+g_2w_{2,B}+g_nw_{3,B}),\\
  a_{2,1} &= g_A (g_1g_2w_{1,A}w_{2,A}+g_1g_3w_{1,A}w_{3,A}+g_2g_3w_{2,A}w_{3,A})\\
          &\qquad +g_B (g_1g_2w_{1,B}w_{2,B}+g_1g_3w_{1,B}w_{3,B}+g_2g_3w_{2,B}w_{3,B}), \\
  a_{3,1} &= g_A g_1 g_2 g_3 w_{1,A} w_{2,A} w_{3,A} + g_B g_1 g_2 g_3 w_{1,B} w_{2,B} w_{3,B},\\
  a_{1,2} &= g_Ag_B (g_1 w_{1,A}w_{1,B}+ g_2 w_{2,A}w_{2,B} +g_3 w_{3,A}w_{3,B}),\\
  a_{2,2} &= g_Ag_B (g_1 g_2 w_{1,A}w_{1,B} w_{2,A}w_{2,B}+ g_1 g_3 w_{1,A}w_{1,B} w_{3,A}w_{3,B} + g_2 g_3 w_{2,A}w_{2,B} w_{3,A}w_{3,B}), \\
  a_{3,2} &= g_Ag_B g_1 g_2 g_3 w_{1,A}w_{1,B}w_{2,A}w_{2,B} w_{3,A}w_{3,B}.
\end{align*}
Choosing
\begin{align*}
  &g_1 = 2, \quad g_2 = 3, \quad g_3 = 5, \quad g_A = 11, \quad g_B = 13, \\
  &a_{1,1} = 71, \quad a_{2,1}=73, \quad a_{3,1}=79, \quad a_{1,2}=101,\quad  a_{2,2}=103, \quad a_{3,2}=107,
\end{align*}
the system then simplifies to
\begin{align*}
  71 &= 11 (2w_{1,A}+3w_{2,A}+5w_{3,A})+ 13 (2w_{1,B}+3w_{2,B}+5w_{3,B}),\\
  73 &= 11 (6w_{1,A}w_{2,A}+10w_{1,A}w_{3,A}+15w_{2,A}w_{3,A})\\
          &\qquad +13 (6w_{1,B}w_{2,B}+10w_{1,B}w_{3,B}+15w_{2,B}w_{3,B}), \\
  79 &= 330 w_{1,A} w_{2,A} w_{3,A} + 390 w_{1,B} w_{2,B} w_{3,B},\\
  101 &= 143 (2 w_{1,A}w_{1,B}+ 3 w_{2,A}w_{2,B} +5 w_{3,A}w_{3,B}),\\
  103 &= 143 (6 w_{1,A}w_{1,B} w_{2,A}w_{2,B}+ 10 w_{1,A}w_{1,B} w_{3,A}w_{3,B} + 15 w_{2,A}w_{2,B} w_{3,A}w_{3,B}), \\
  107 &= 4290 w_{1,A}w_{1,B}w_{2,A}w_{2,B} w_{3,A}w_{3,B}.
\end{align*}
Using \textsc{bertini}, we see that it has 72 roots, all of which are non-real and simple, 12 of which have strictly positive real component. Figure~\ref{fig:solution32} shows a pair of complex conjugate solutions with lexicographically largest real part. The ordering on the variables is $w_{1,A}$, $w_{1,B}$, $w_{2,A}$, $w_{2,B}$, $w_{3,A}$, $w_{3,B}$, so that the first lines corresponds to the real and imaginary part of $w_{1,A}$.

\begin{figure}[ht]
  \begin{tabular}{c}
\begin{lstlisting}[basicstyle=\scriptsize]
0.733175658157242746563365475886e0 -0.525124949875722284087132912860e0
0.261871644858051814095937009460e0 -0.442937573433368024261374882297e0
0.175843060588207991285330446970e0 0.331651692189479173140340537837e0
0.691545419943605655254684211883e0 0.448271194339440245246537764856e0
0.154572711574605557323732099093e0 -0.483019502079697923238189339927e0
0.104413557544869576753196043395e0 0.326278707683474805854874986793e0

0.733175658157242746563365475898e0 0.525124949875722284087132912872e0
0.261871644858051814095937009460e0 0.442937573433368024261374882309e0
0.175843060588207991285330446974e0 -0.331651692189479173140340537843e0
0.691545419943605655254684211896e0 -0.448271194339440245246537764868e0
0.154572711574605557323732099096e0 0.483019502079697923238189339940e0
0.104413557544869576753196043398e0 -0.326278707683474805854874986799e0
\end{lstlisting}
  \end{tabular}\vspace{-0.25cm}
\caption{Two complex conjugate solutions for $(3,2)$.}\label{fig:solution32}
\end{figure}

\subsection{Mixed volumes}\label{sec:mixedVolumes}

The Newton polytope of a polynomial is the convex hull of all exponent vectors of all monomials with non-zero coefficient. Given a polynomial system $f_1,\ldots,f_N$ in $N$ variables, the \emph{mixed volume} of their Newton polytopes is a number that equals the number of roots provided the non-zero coefficients are generic. This is known as the Bernstein-Khovanskii-Kushnirenko Theorem \cite{Bernstein75}.

Figure \ref{fig:mixedVolumeAndMore} shows a table with the mixed volume for various $(n_1,n_2)$ computed using \textsc{gfan}. We see that the number for $(n_1,1)$ and $(n_1,2)$ corresponds with the theoretical results. Sadly, there is no apparent pattern for $(n_1,n_2)$ with $n_2>2$.

Note that there exist criteria on so-called Newton-degeneracy which guarantee that the mixed volume equals the number of roots \cite{BS95}. However, due to the high-dimensionality of the Newton polytopes, these were infeasible to verify for the cases of interest such as $(4,3)$.

\begin{figure}[ht]
  \centering
  \setlength{\tabcolsep}{8pt}
  \begin{tabular}{c|rrrrrr}
      & 1 & 2 & 3    & 4        & 5        & 6 \\[3pt] \hline
    1 & \textcolor{red}{1} & \textcolor{red}{2} & \textcolor{red}{6}    & \textcolor{red}{24}       & \textcolor{red}{120}      & \textcolor{red}{720} \\[3pt]
    2 &   & \textcolor{red}{8} & \textcolor{red}{72}   & \textcolor{red}{1\! 152}     & \textcolor{blue}{28\! 800}    & 1\! 036\! 800 \\[3pt]
    3 &   &   & \textcolor{red}{1\! 944} & \textcolor{blue}{162\! 432}   & 24\! 624\! 000 & 1\! 349\! 713\! 408 \\[3pt]
    4 &   &   &      & 52\! 862\! 976 & -        & - \\[3pt]
    5 &   &   &      &          & -        & - \\[3pt]
    6 &   &   &      &          &          & -
  \end{tabular}\vspace{-0.25cm}
  \caption{mixed volumes of the Newton polytopes of (\ref{oursystem})}
  \label{fig:mixedVolumeAndMore}
\end{figure}

\subsection{Counting solutions with multiplicity}

Given a zero-dimensional polynomial ideal $I\unlhd \CC[x]$, the dimension of $\CC[x]/I$ as a $\CC$-vector space equals the number of solutions counted with multiplicity. It can be easily read off any Gr\"obner basis, but computing the Gr\"obner basis itself is a highly challenging task \cite[Section 1.8.5]{GP08}. In Figure~\ref{fig:mixedVolumeAndMore}, red numbers mark all cases for which Gr\"obner bases were computeable in \textsc{Singular}. The respective vector space dimensions (computed using the \textsc{Singular} command \texttt{vdim}) all coincided with the mixed volume.

\subsection{Explicit solutions for $(5,2)$ and $(4,3)$}

For the cases $(5,2)$ and $(4,3)$, highlighted blue in Figure~\ref{fig:mixedVolumeAndMore}, we also tried to compute explicit roots using \textsc{bertini}. However, numerical instabilities arose in both cases during the computation, so that the roots computed are not complete.

For $(5,2)$ we obtained $28737$ roots, $63$ short or $99.8\%$ of the proven $28800$ roots. For $(4,3)$ we obtained $156966$ roots, $5466$ short or $97\%$ of the conjectured $162432$ roots.
\section{Open questions}\label{sec:summaryAndOpenQuestion}

We close with three open questions.

\begin{question}
What is the number of solutions for $(n_{1},n_{2})$?
\end{question}

For binding polynomials of bidegree $(n,1)$ and $(n,2)$, the number of decoupled molecules is given by relatively simple expressions. Assuming that the mixed volumes of the Newton polytopes equals the number of solutions, Table~\ref{fig:mixedVolumeAndMore} indicates a more complicated pattern in the number of decoupled molecules for $(n,3)$. The smallest interesting example is the case $(4,3)$ for which we conjecture that the number equals $162432$.

\begin{question}
 How many solutions with real, positive values for $g_{i}$ and $w_{i,j}$ exist?
\end{question}

For univariate binding polynomials, the existence of complex roots suggests that the system does strongly interact and cannot be represented by a real decoupled system. In particular it is an indicator for ``cooperativity'' \cite{Martini16model}. It is neither clear how this concept can be translated to decoupled molecules for two types of ligands nor which characteristic different decoupled molecules share. To develop an understanding, it would be helpful to determine the number of real, positive solutions for small examples.

\begin{question}
  Find an algorithm to compute the minimal interaction energy that a molecule with prescribed binding polynomial has.
\end{question}

For univariate binding polynomials, a quantitative measure for ``cooperativity'' is mapping the polynomial to the minimal interaction energy which is required to generate it \cite{martini17measure}. In more detail, the \textit{norm} of a molecule is the product of all the \textit{absolutes} of its interaction energies,
\begin{align*}
|\M| = \prod |w_{i,j}|, \text{ where } |w| := \max\left(w,w^{-1}\right),
\end{align*}
while the \textit{norm} of a polynomial $\Phi$ is the minimal norm of all molecules that give rise to this polynomial. How can we calculate $|\Phi|$? It would be interesting to investigate whether the machinery that has been developed for Euclidean distance degree is applicable in our setting \cite{DHOST16}.

\bibliography{bp}
\bibliographystyle{plain}

\end{document}